\documentclass[sigconf]{acmart}

\copyrightyear{2020}
\acmYear{2020}
\setcopyright{acmcopyright}
\acmConference[HSCC '20]{23rd ACM International Conference on Hybrid Systems: Computation and Control}{April 22--24, 2020}{Sydney, NSW, Australia}
\acmBooktitle{23rd ACM International Conference on Hybrid Systems: Computation and Control (HSCC '20), April 22--24, 2020, Sydney, NSW, Australia}
\acmPrice{15.00}
\acmDOI{10.1145/3365365.3382203}
\acmISBN{978-1-4503-7018-9/20/04}


\usepackage{graphicx}
\usepackage{habbas-macros}
\usepackage{enumerate}
\usepackage{paralist}
\usepackage{wrapfig}
\usepackage{ltl}
\def\Choicear{Choice_\alpha^{root}}
\def\Optimalar{Optimal_\alpha^{root}}

\def\State{S\!tate}
\def\vnl{\ell_n}
\def\vnu{u_n}

\def \wbuntil{\overline{\until}_N}

\newtheorem{definition}{Definition}
\newtheorem{theorem}{Theorem}
\newtheorem{proposition}{Proposition}
\newtheorem{lemma}{Lemma}
\graphicspath{figs/}
\begin{document}
\title{A Deontic Logic Analysis of Autonomous Systems' Safety}
%
%

%
%

\author{Colin Shea-Blymyer}
\email{sheablyc@oregonstate.edu}
\affiliation{%
  \institution{Oregon State University}
}

\author{Houssam Abbas}
\email{houssam.abbas@oregonstate.edu}
\affiliation{\institution{Oregon State University}}

\begin{abstract}
We consider the pressing question of how to model, verify, and ensure that autonomous systems meet certain \textit{obligations} (like the obligation to respect traffic laws), and refrain from impermissible behavior (like recklessly changing lanes).
Temporal logics are heavily used in autonomous system design;
however, as we illustrate here, temporal (alethic) logics alone are inappropriate for reasoning about obligations of autonomous systems.
This paper proposes the use of Dominance Act Utilitarianism (DAU), a deontic logic of agency, to encode and reason about obligations of autonomous systems.
We use DAU to analyze Intel's Responsibility-Sensitive Safety (RSS) proposal as a real-world case study.
We demonstrate that DAU can express well-posed RSS rules, 
formally derive undesirable consequences of these rules, 
illustrate how DAU could help design systems that have specific obligations,
and how to model-check DAU obligations.
\end{abstract}

\maketitle              


\section{Obligations, Permissions and Norms For Autonomous Vehicles}
\label{sec:intro}
There is now a realistic prospect that Autonomous ground Vehicles (AVs) will be deployed on public roads in the next few years, with Waymo already charging customers for self-driving taxi in Arizona~\cite{Hawkins.2019}. 
While companies produce `event reports' to regulators, there is a worrying sparsity of rigorous verification methods, and of external independent assessment, of the vehicles' performance.
The most pressing issue is that of verifying safety.
So far, the vast majority of the work in formal verification of AVs used the tools of \textit{alethic temporal logic} (like Linear~\cite{Pnueli77sfcs} or Metric Temporal Logic~\cite{Koymans90}) to express behavioral specifications of system models.
Alethic logic is the logic of \textit{necessity and possibility}:
for example, if $p$ is a predicate, $\always p$ says that $p$ is true in every accessible world - that is, $p$ is necessary.
Possibility is then formalized as $\eventually p \defeq \neg \always \neg p$: saying that $p$ is possible is the same as saying that it is not the case that $\neg p$ is necessary.
And so on.
The best known instantiation of this in Verification is LTL~\cite{MannaP92}, in which an accessible world is a moment in the (linear) future.
Thus $\always p$ formalizes `$p$ is true in every future moment', and $\eventually p$ formalizes `$p$ is true in some future moment'.

It is, however, equally important to think in terms of \textit{obligations and permissions} of the autonomous system:
for instance, we may wish to say that `It is obligatory for the AV to not rear-end a car', or `It is permissible to drive on the shoulder if the car ahead brakes suddenly'.
Obligations, permissions and prohibitions are also pervasive when discussing ethical questions: what should the AV do when faced with two equally unsavory but inevitable alternatives?
Obligations and permissions are collectively called \textit{norms} and statements about them are called normative statements.
A prominent example of a proposed normative system for Autonomous Vehicles (AVs) is Intel's Responsibility-Sensitive Safety (RSS)~\cite{RSSv6}, which states what the AV should and should not do to avoid accidents.
\textit{It is essential to logically formalize proposed norms for autonomous systems to enable automatic reasoning about their logical consistency, consequences, and automate system design}.
While all current work in AV verification and testing uses temporal logics~\cite{YaghoubiF19MLtesting}, which are types of alethic logic, it has been understood for over 70 years that \textit{the logic of norms is different from that of necessity}~\cite{dlStanford}: applying alethic logic rules to normative statements leads to conclusions that are intuitively paradoxical or undesirable.
Consider the following statements:
\begin{enumerate}[A.]
	\item The car will eventually change lanes: this is a statement about possibility. It says nothing about whether the car plays an active role in the lane change (e.g., perhaps it will hit a slippery road patch).
	\item The car sees to it that it changes lanes: this is a statement about agency. 
	It tells us that the car is an active agent in the lane change, or is choosing to change lanes.
	\item The car can change lanes: this is a statement about ability. The car might be able to do something, but have no `choice' or agency in the matter.
	\item The car ought to change lanes: this is a statement about obligation, a concept not captured in the first three statements.
\end{enumerate}
These are qualitatively different statements and there is no a priori equivalence between any two of them.
The logic we adopt should reflect this: its operators and inference rules should model these aspects.
Alethic logics like LTL cannot do so.

We now give a simple but fundamental example, drawn from~\cite{dlStanford}, illustrating this point.
(In Section~\ref{sec:dau} we give an AV-specific example.)
One might be tempted to formalize obligation using the necessity operator $\always$: that is, formalize `The AV should stay in its lane' by $\always \texttt{stay-in-lane}$.
However, in alethic logic, $\always p \implies p$: if $p$ is necessarily true then it is true.
If we interpret $\always$ as obligation this reads as $\mathbf{Obligatory}~p \implies p$: this is clearly non-sensical because agents sometimes violate their obligations so some obligatory things are not true.
%
%
This leads us to a major question in studying obligations: the automatic derivation of what an agent should do when some primary obligations are violated. 
I.e. we wish to study statements of the form $\mathbf{Obligatory}p \land \neg p \implies ...$.
 This is simply impossible in pure alethic logic, since $\always p \land \neg p \implies q$ is trivially true for \textit{any} $p$ and $q$.
Thus alethic logics (including common temporal logics like LTL, MTL or CTL~\cite{ClarkeGP99}) are not appropriate, on their own, for automatic reasoning about norms.

\textit{Deontic logic}~\cite{DLHandbook} has been developed specifically to reason about normative statements, starting with von Wright~\cite{vonWright51DL}. 
It is widely used in contract law, including software contracts.
There are many flavors of deontic logic~\cite{McNamaraChapter}.
In this paper, we adopt Dominance Act Utilitarianism (DAU) developed by Horty~\cite{Horty01DLAgency} because it explicitly models all four aspects above: necessity, agency, ability and obligation.
It includes a temporal logic as a component so we can describe temporal behaviors essential to system design, 
and it uses branching time, essential for modeling uncontrollable environments.

To assess whether DAU is appropriate for reasoning about the norms of autonomous systems, we formalize a subset of Intel's Responsibility-Sensitive Safety, or RSS, in DAU. 
RSS proposes a set of norms or rules that, if followed by all cars in traffic, would lead to zero accidents~\cite{RSSv6}.
The RSS proposal is expressed in the language of continuous-time dynamical systems and ordinary differential equations, but the rules to be followed are not formalized logically, so it is not possible to reason about them.
This work integrates formal methods in AV design by complementing the dynamical equations-based presentation of RSS in~\cite{RSSv6} with a deontic logic formalism.
We formalize RSS in DAU, which achieves three purposes:
first, it demonstrates the usefulness of DAU in a real use case, namely, the analysis of a safety proposal by a major player in autonomous driving technology.
Second, it realizes a necessary first step towards automated system design.
Finally, it allows a systematic discovery of implicit assumptions, and undesirable consequences of any such proposals.
A framework to do this is still missing from the literature. 
Our contributions in this paper are to:
\begin{enumerate}
	\item formalize the normative system of RSS in DAU, to highlight the subtle decisions that need to be made when developing a \textit{rigorous} safety specification (Section~\ref{sec:rss formalization});
	\item partially infer the system structure using the DAU formalization (Section~\ref{sec:analysis inferring policies});
	\item derive undesirable consequences of the RSS norms, pointing the way to further necessary refinements of the norms (Section~\ref{sec:other accidents}); and 
	\item develop a model-checking algorithm of DAU specifications that allows to establish whether a system has a given obligation or not (Section~\ref{sec:model checking}).
\end{enumerate}

\subsection{Related work}
\label{sec:why this dl}
There is a wide variety of deontic logics, tailored to different ends~\cite{DLHandbook}. 
Standard Deontic Logic has many well-known paradoxes~\cite{McNamaraChapter}, which have spurred the proposal of alternatives to remedy them.
Some variations are commonly used to specify legal and software contracts as in~\cite{Prisacariu2012CL}.
Various attempts were made to integrate deontic logic with temporal modalities 
(e.g.,~\cite{Giordano13TDLASP} and~\cite{Raimondi04automaticverification}).
Decision procedures exist for some logics, 
like the checker in~\cite{Lomuscio2017}.
Gerdes et al.~\cite{Gerdes2015} have compared a deontological approach to AV design with a consequentialist approach by formalizing them as an optimal control problem.
Rizaldi et al.~\cite{Althoff15TrafficRules} formalize six traffic rules in Higher Order Logic to be passed to an interactive theorem prover.
As it is our goal to logically analyze normative safety rules and use them in system design,~\cite{Gerdes2015} and~\cite{Althoff15TrafficRules} present directions of investigation that are orthogonal to ours. 
Alternating-time Temporal Logic (ATL) was proposed in~\cite{AlurHK02ATL} and extended in~\cite{vanderHoek2003ATEL} to reason about groups of agents. ATL seems to use sure-thing reasoning, like DAU (see Section~\ref{sec:dau}), but does not natively support a notion of obligation. 
The RSS proposal itself~\cite{RSSv6} uses a point mass dynamical model 
to derive definitions of minimum safe distances between two cars. 
It also proposes motion planning policies to avoid accidents; e.g., if the car ahead hits maximum brakes, then the following car should hit maximum brakes within a delay $\tau$, and so on.
The RSS rules are not formalized in any logic in~\cite{RSSv6}, nor are its logical consequences examined. 
This paper leverages DAU's formulation of agency~\cite[3.3]{Horty01DLAgency} to formalize well-posed RSS rules and analyze their implications. DAU further distinguishes itself through its distinction between what ought to be the case and what an agent ought to do~\cite[3.3]{Horty01DLAgency}. A related formulation to DAU is found in~\cite{Brunel08Propagation}.

\section{Dominance Act Utilitarianism}
\label{sec:dau}
\subsection{A deontic logic over branching time}
\label{sec:horty's logic}
This section summarizes the main aspects of DAU developed in~\cite{Horty01DLAgency}, starting with classical branching time models.
Let \textit{Tree} be a set of \textit{moments} with an irreflexive, transitive ordering relation $<$ such that for any three moments $m_1,m_2,m_3$ in $\Tree$, if $m_1<m_3$ and $m_2 < m_3$ then either $m_1<m_2$ or $m_2 < m_1$.
There is a unique \textit{root moment} of the tree satisfying $root < m'$ for all $m'\neq root$.
A \textit{history} is a maximal linearly ordered set of moments from \textit{Tree}: intuitively, it is a branch of the tree that extends infinitely.
Given a moment $m \in$ \textit{Tree}, the set of histories that go through $m$ is $H_m \defeq  \{ h \such m \in h\}$.
See Fig.~\ref{fig:choices}
We will frequently refer to moment/history pairs $m/h$, where $m \in$ \textit{Tree} and $h \in H_m$.
\begin{definition}\cite[Def. 2.2]{Horty01DLAgency}
	\label{def:frame and model}
	With $AP$ a set of atomic propositions,
	a \emph{branching time model} is a tuple $\Model=(\Tree, <,v)$ where \textit{Tree} is a tree of moments with ordering $<$ and $v$ is a function that maps $m/h$ pairs in $\Model$ to sets of atomic propositions from $2^{AP}$.
\end{definition}
A branching time model can be seen as the result of executing a non-deterministic automaton that models all agents in the system. 
While we will frequently speak of one agent's obligations for simplicity, the reader should keep in mind that a model $\Model$ can represent the possible evolutions of several agents.
\begin{figure*}[t]
	\centering
	\includegraphics[height=2.25in,keepaspectratio]{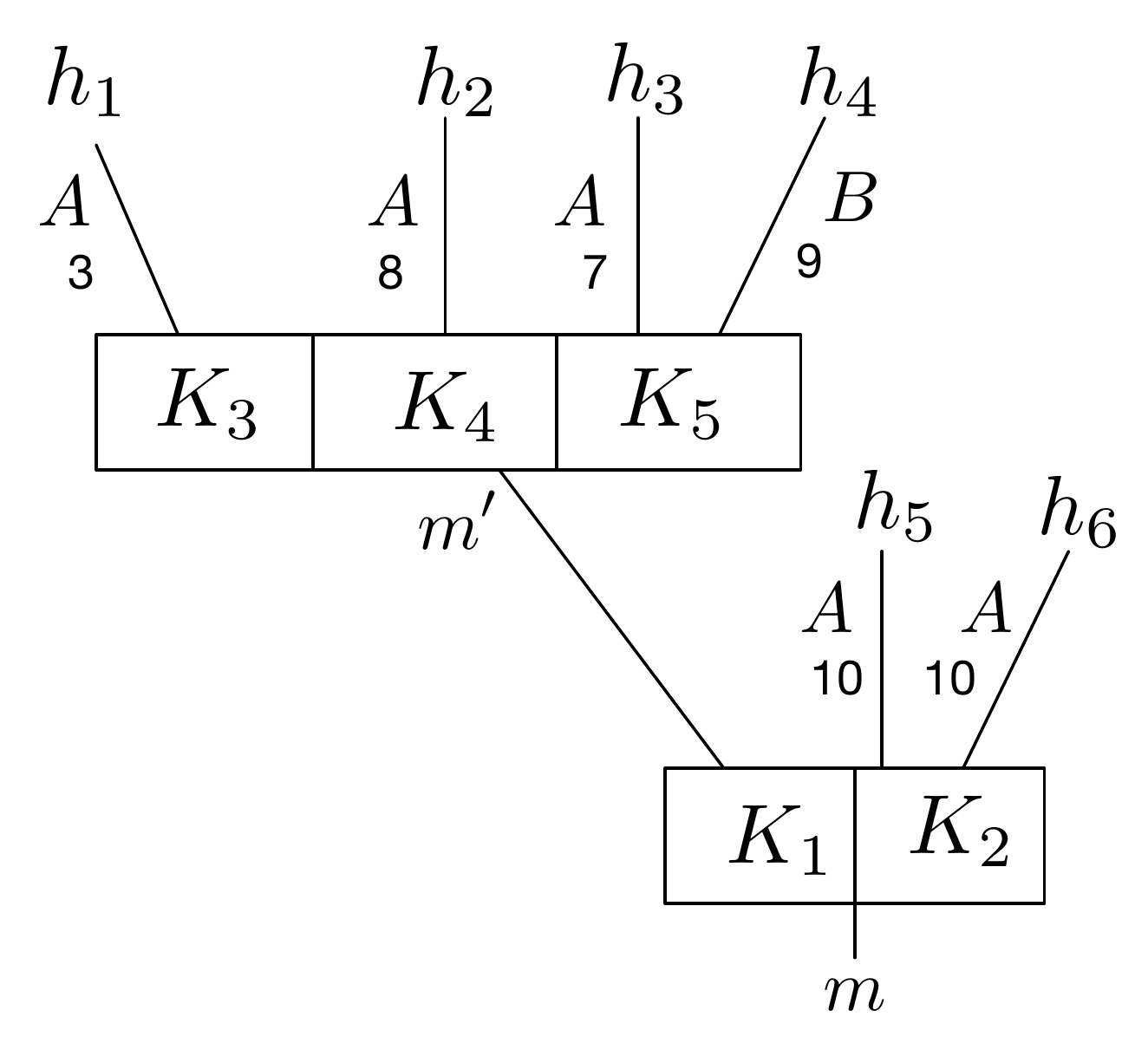}
	\caption{
	A utilitarian stit model for an agent $\alpha$, showing moments $m<m'$ with sets of histories $H_{m} = \{h_1,\ldots, h_6\}$ and $H_{m'}=\{h_1,\ldots, h_4\}$. 
		Each moment is marked with the actions available to $\alpha$ at that moment: $\Choiceam = \{K_1,K_2\}$ and $Choice_\alpha^{m'} = \{K_3,K_4,K_5\}$. 
	Action $K_2=\{h_5,h_6\}$ and $K_4=\{h_2\}$. 
	Each history is marked with the formula(s) that it satisfies and with its value $\Value(h)$, e.g., $h_1$ satisfies $A$ and has value 3.
	$m/h_5 \models \cstitaa$ since $\Choiceam(h_5) = K_2$, and both $h_5$ and $h_6$ satisfy $A$. 
	On the other hand, $m/h_1 \nvDash \cstit{\alpha}{A}$ since $\Choiceam(h_1) = K_1= \{h_1,h_2,h_3,h_4\}$ and $h_4$ does not satisfy $A$. 
	$\Optimalam = \{K_2\}$ so $m/h_5 \models \Ostit{\alpha}{A}$. 
	$\Optimal{\alpha}{m'} = \{K_4,K_5\}$ and so $\alpha$ has no obligations at $m'$ since there is no formula $\formula$ s.t. $|\formula|_{m'} \supseteq K_4 \cup K_5$ (See Def.~\ref{def:ought semantics}). 
	Finally, $m/h_5 \models \dstit{\alpha}{A}$ because $K_2 \subset |A|_m$ and $H_m \neq |A|_m=\{h_1,h_2,h_3,h_5,h_6\}$.
	} 
	\label{fig:choices}
\end{figure*}

We will use \CTLs~as the tense logic on branching time models - see~\cite{ClarkeGP99} for details.\footnote{The development of DAU in~\cite{Horty01DLAgency} uses a restricted temporal logic, but that is immaterial here.}
\CTLs~ includes computational tree logic (CTL) and linear temporal logic (LTL), and has become widely used in model checking. \CTLs~ can produce sentences like 
\(\formula := \exists X (p) \land \forall \eventually \always (p)\) which can be interpreted as `there exists a path where $p$ holds at the next state, and all paths will eventually always satisfy $p$'.
\CTLs~allows us to formalize the temporal evolution of events along a given history $h$  (e.g., $\eventually \formula$),
and quantify over histories passing through a moment $m$ (e.g., $\forall \formula$ meaning `for all histories, $\formula$ holds').
In this paper, to retain a uniform satisfaction relation like~\cite{Horty01DLAgency}, we will speak of formulas holding or not at an $m/h$ pair: for a pair $m/h$ in a model $\Model$, we write $\Model, m/h \models \formula$, where it is always the case that $h\in H_m$.
There should be no confusion as a \CTLs~path formula is evaluated along $h$ and a state formula is evaluated at $m$.

A formula $\formula$ is identified at moment $m$ with the set of histories where it holds 
\begin{equation}
\label{eq:Amm}
|\formula|_m^\Model \defeq \{h \in H_m \such \Model, m/h \models \formula\}
\end{equation}
Where there's no risk of ambiguity, we drop $\Model$ from the notation, writing $|\formula|_m$, etc.


The rest of this section is dedicated to the exposition of the properly deontic aspects of DAU.

\paragraph{Choice}
Let $\Agent$ be a set of agents, which represent, for example, the cars in traffic. 
Consider an agent $\alpha \in \Agent$ and a given model $\Model$.
Then at every moment $m$, $\alpha$ is faced with a choice of actions which we denote by $\Choiceam$.
Intuitively, an action causes some histories from $H_m$ to no longer be realizable, while others still are. 
Thus we can identify each action $K \in \Choiceam$ with the set of histories that are still realizable after taking the action, and we may write $K \subseteq H_m$.
See moments and actions in Fig.~\ref{fig:choices}.
$\Choiceam$ must obey certain constraints which we relegate to Appendix~\ref{apx:choice}. 

\paragraph{Agency}
Agency is defined via the Chellas `sees to it' operator $cstit$, named in honor of Brian Chellas who introduced an analogous operator in \cite{chellas1968logical}.
(Saying `John sees to it that the window is open' means that John ensures the window is open). 
Intuitively, an agent sees to it that $A$ by taking action $K$ at $m/h$ iff, whatever other history $h'$ could've resulted from the action, $A$ is true at $m/h'$ as well. 
Thus, the non-determinism does not prevent $\alpha$ from achieving $A$.
Let $\Choiceam(h)$ be the unique action that contains $h$.
In Fig.~\ref{fig:choices} $\Choiceam(h_1) = K_1 =  \{h_1,h_2,h_3,h_4\}$.
\begin{definition}[Chellas stit]~\cite[Def. 2.7]{Horty01DLAgency}
	\label{def:cstit}
	With agent $\alpha$ and formula $\formula$
	\[\Model,m/h \models [\alpha\,cstit:\formula] \text{ iff } \Choiceam(h) \subseteq |\formula|_m^\Model\]
	\end{definition}
See Fig.~\ref{fig:choices}.
We also define a \textit{deliberative stit} operator, which captures the notion that an agent can only truly be said to do something if it also has the choice of not doing it. 

 \begin{definition}[Deliberative stit]~\cite[Def. 2.8]{Horty01DLAgency}
	\label{def:dstit}
		With $\alpha$ and $\formula$ as before, 
		\[\Model,m/h \models [\alpha~dstit\!:\formula] \text{ iff } \Choiceam(h) \subseteq |\formula|_m^\Model \text{ and } |\formula|_m^\Model \neq H_m\]
\end{definition}
Thus $\dstit{\alpha}{A}$ iff some histories don't satisfy $A$ but $\alpha$'s choice ensures $A$.
See Fig.~\ref{fig:choices}.
The operators $cstit$ and $dstit$ are not interchangeable and they fulfill complementary roles. 
This paper focuses on obligation statements of the following form.
\begin{definition}[Obligations]
	\label{def:obligation structure}
	Let $\alpha$ be an agent.
	An \emph{obligation} $A$ is either a \CTLs~formula, or a statement of the form $\dstit{\alpha}{\formula}$ or $~\neg \dstit{\alpha}{\formula}$ where $\formula$ is a \CTLs~formula.	
\end{definition}
Like Eq.~\eqref{eq:Amm} for \CTLs~formula, we identify an obligation $A$ at moment $m$ with the set of histories where it holds
\begin{equation}
\label{eq:Amm2}
|A|_m^\Model \defeq \{h \in H_m \such \Model, m/h \models A\}
\end{equation}

Obligations can be used in \textit{stit} formulations, similarly to formulas in definitions \ref{def:cstit} and \ref{def:dstit}:
	\[\Model,m/h \models [\alpha\,cstit:A] \text{ iff } \Choiceam(h) \subseteq |A|_m^\Model\]
	and
	\[\Model,m/h \models [\alpha~dstit\!:A] \text{ iff } \Choiceam(h) \subseteq |A|_m^\Model \text{ and } |A|_m^\Model \neq H_m\]

\paragraph{Optimal actions.}
To speak of an agent's obligations, we will need to speak of `optimal actions', those actions that bring about an ideal state of affairs.
We make the simplifying assumption that all agents in the system collaborate to achieve a common goal.
This is consistent with the RSS assumption that all agents are following the same rules to avoid collisions anywhere in traffic.
Let $\Value: H_{root}\rightarrow \Re$ be a \textit{value function} that maps histories of $\Model$ to \textit{utility values} from the real line $\Re$.
This value represents the utility associated by all the agents to this common history.

\begin{definition}
	\label{def:utilitarian stit frame}
	A \emph{utilitarian stit frame} is a tuple 
	$(\Tree, <, \Agent, $
	$\Choicemap, \Value)$ where 
	\textit{Tree} and $<$ are as in branching time frames, 
	$\Agent$ is a set of agents, 
	$\Choicemap$ is a choice mapping (which is specialized as $\Choiceam$ for each agent and moment),
	and $\Value$ is a value function.
	A \emph{utilitarian stit model} is a model based on a utilitarian stit frame.
	If $\Choiceam$ is finite for every $\alpha\in \Agent$ and $m$, the model is said to be \emph{finite-choice}.
\end{definition}
All models in what follows are finite-choice utilitarian stit models. 
Given two sets of histories $X$ and $Y$ , we order them as 
\begin{equation}
\label{eq:proposition dominance}
X\leq Y \text{ iff } \Value(h) \leq \Value(h') \quad \forall~h\in X, h' \in Y
\end{equation}
%
Let 
$\State_{\alpha}^m \defeq \Choice{\Agent \setminus \{\alpha\}}{m}$
be the set of \textit{background states} against which $\alpha$'s decisions are to be evaluated.
These are other agents' independent actions.
Given two actions $K,K'$ in $\Choiceam$, 
$K\preceq K' \text{ iff } K \cap S \leq K'\cap S \textrm{ for all } S \in \State_{\alpha}^m$.
That is, $K'$ dominates $K$ iff it is preferable to it regardless of what the other agents do (known as \textit{sure-thing reasoning}).
Strict inequalities are naturally defined.
Optimal actions are given by \cite{Horty01DLAgency}
\begin{equation}
\label{eq:optimalam}
\Optimalam \defeq \{K \in \Choiceam \such \not\exists K' \in \Choiceam~.~K\prec K'\}
\end{equation}
$\Optimalam$ is non-empty in finite-choice utilitarian stit models~\cite[Thm. 4.10]{Horty01DLAgency}.

\paragraph{Dominance Ought}
Intuitively we will want to say that at moment $m$, agent $\alpha$ \textit{ought to see to it} that $A$ iff $A$ is a necessary condition of all the histories considered ideal at moment $m$.
This is formalized in the following \textit{dominance Ought operator}, which is pronounced ``$\alpha$ ought to see to it that $A$ holds''.
\begin{definition}[Dominance ought]
	\label{def:ought semantics}
With $\alpha$ an agent and $A$ an obligation in a model $\Model$, 
\[\Model,m/h \models \Ostit{\alpha}{A} \text{ iff } K \subseteq |A|_m^{\Model}~~\text{ for all } K \in \Optimalam\]
\end{definition}
See Fig.~\ref{fig:choices} for examples.
If $K \subseteq |A|_m$ we say that $K$ guarantees $A$.
Note that the dominance Ought is only defined with the $cstit$ operator and not $dstit$; this is because it leads to a simpler logic.
The dominance ought satisfies a number of pleasing logical properties; we refer the reader to~\cite[Ch. 4]{Horty01DLAgency}.

\paragraph{Conditional obligation}
It is often necessary to say that an obligation is imposed only under certain conditions. 
Where $A$ and $B$ are obligations, the statement
\begin{equation}
\label{eq:conditional ought expression}
\Model, m/h \models \OstitC{\alpha}{A}{B}
\end{equation}
expresses that $\alpha$ ought to see to it that $A$, under the condition that $B$ holds.
\begin{definition}[Conditional ought]
	\label{def:conditional ought semantics}
With $\alpha$ an agent and $A$, $B$ as obligations in a model $\Model$, 
\[\Model,m/h \models \OstitC{\alpha}{A}{B} \text{ iff }\]
\[K \subseteq |A|_m^{\Model}~~\text{ for all } K \in \Optimalam /|B|_m^{\Model}\]
\end{definition} 
where \(\Optimalam /B\) ($\alpha$'s optimal actions under the condition $B$) is the set of actions available to $\alpha$ that are optimal if we ignore $B$-violating histories \cite{Horty01DLAgency}.

We note that conditional obligation is \textit{not} the same as $B \implies \Ostit{\alpha}{A}$.\footnote{This is not a well-formed DAU expression, but we can extend the logic to give this expression its natural definition as $\neg B \lor \Ostit{\alpha}{A}$.}
Conditional obligation only considers $B$-guaranteeing dominating histories, while this latter formula still considers all optimal actions, not only those that guarantee the truth of $B$.

\paragraph{Syntax}
We now summarize the syntax of DAU statements.
Obligations are generated as follows. 
\[A::=\formula~|~[\alpha dstit:A]~|~\neg A\] 
where $\formula \in$ \CTLs, and the semantics of $~[\alpha dstit:A]$ were given in Def.~\ref{def:dstit}.
Ought statements are in one of two forms:
\[\Ostit{\alpha}{A}~\text{or}~\OstitC{\alpha}{A}{B}\] 
where $\alpha$ is an agent and $A$ and $B$ are obligations. 
The semantics were given in Def. \ref{def:ought semantics}.
\subsection{Alethic Logic vs DAU for Analyzing AV Behavior}
\label{sec:analysis alethic vs deontic}
\begin{figure*}[t]
	\centering
	\includegraphics[height=2.25in,keepaspectratio]{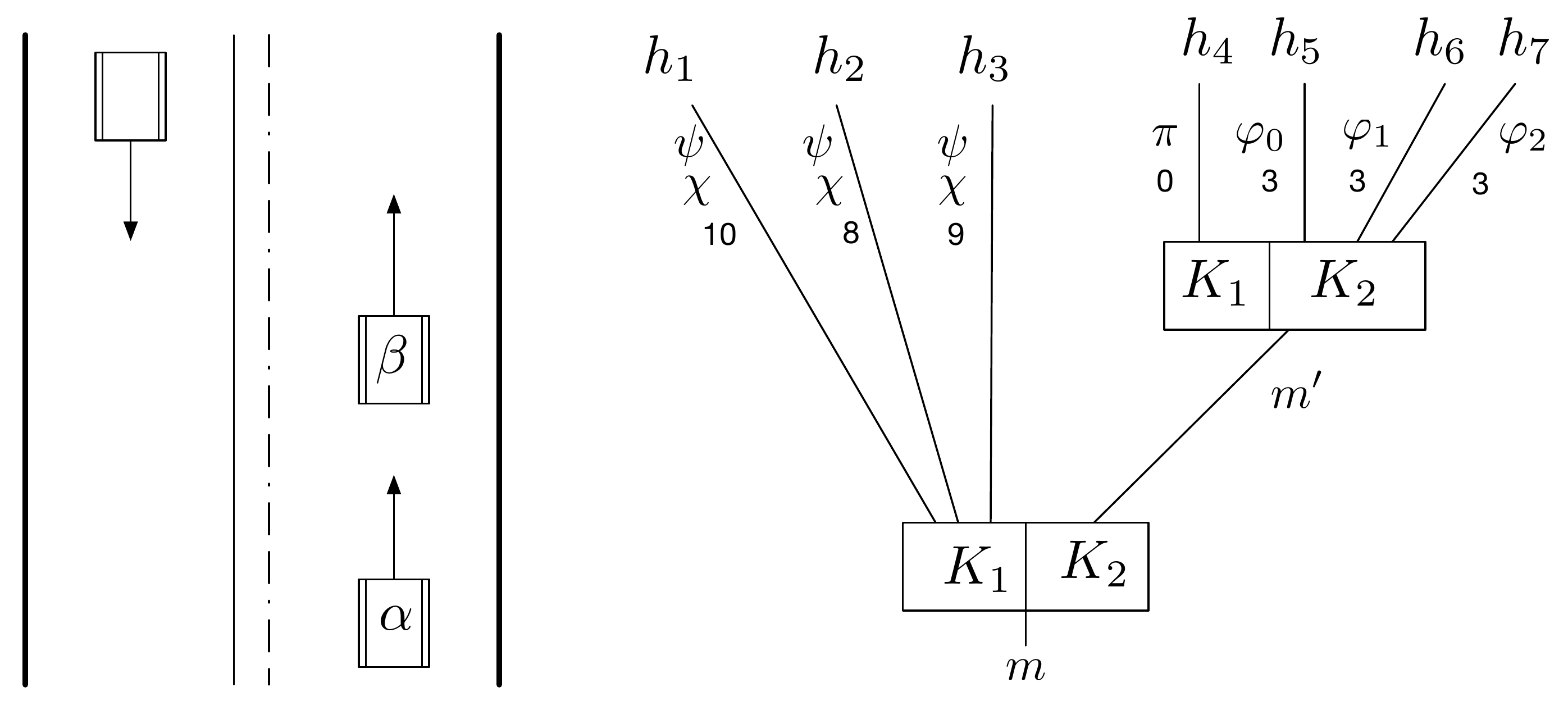}
	\caption{\small Deriving obligations for $\alpha$ from the stit model. $K_1$: stay in lane, $K_2$: change lanes.}
	\label{fig:tree opposite}
\end{figure*}
We now offer an AV-specific example of the advantage that a DAU formalization offers over pure temporal logic.
Specifically, DAU allows deriving obligations over time by construction and in a uniform manner;
attempts to do so using pure temporal logic are unsatisfactory.
Consider the stit model in Fig.~\ref{fig:tree opposite}, which models the situation on the left: agent $\alpha$ could either stay in its lane behind the slower $\beta$ ($K_1$), or pass $\beta$ by going into the opposite lane ($K_2$) and risk a head-on collision.
Every history in $K_1$ is deemed preferable to every history in $K_2$ because $K_1$ eliminates the risk of collision, so we assign history values accordingly, as shown.
If the agent does $K_2$, then it needs to get back into its lane.
Thus at $m'$, every history in $K_2$ is preferable to every history in $K_1$, and this is reflected in the values.
Naturally, the histories in $K_1$ at $m$ satisfy $\psi\defeq \forall (\neg p)$ ($\alpha$ does not pass, i.e., does not change lanes), 
those in $K_1$ at $m'$ satisfy $\pi = \forall \eventually \texttt{Collision}$ (since $\alpha$ remains in the opposite lane in this case),
and those in $K_2$ at $m'$ satisfy $\varphi_t \defeq \eventually_{[0:t]} p $, $t=0,1,2$, which says that $\alpha$
changes lanes in at most $t$ time steps ($\eventually_{[n:m]}\formula\defeq \Next^n\formula\lor \Next^{n+1} \formula\ldots\lor \Next^m \formula$ and $\Next^{t}$ is $\Next$ repeated $t$ times).
Moreover, suppose $K_1$ histories at $m$ satisfy some arbitrary formula $\chi$.
The following obligations are then automatically derived from the stit model\footnote{In DAU, $\Ostit{\alpha}{\phi} \land \Ostit{\alpha}{\psi} \text{ is equivalent to } \Ostit{\alpha}{\phi \land \psi}$}:
\begin{eqnarray}
\text{At }m &,& \Ostit{\alpha}{\psi \land \chi} \label{eq:don't pass}
\\
\text{At }m' &,& \Ostit{\alpha}{\varphi_2}~(\text{since } \varphi_2 \text{ is true if } \varphi_0 \text{ or } \varphi_1 \text{ are}) \label{eq:pass}
\end{eqnarray}
Thus it emerges that at $m$, $\alpha$ ought to not change lanes. 
Also at $m$, $\alpha$ ought to see to it that $\chi$ - which may have nothing to do with how the values were assigned to the histories.
E.g., $\chi$ might constrain the motor's energy consumption; it is nonetheless an obligation because it's a necessary condition for achieving an optimal history.
If the agent violates \eqref{eq:don't pass} at $m$ by doing $K_2$, then automatically the model yields that its obligation at $m'$ is \eqref{eq:pass}.
As explained in the Introduction, such generation of new obligations is not possible in pure temporal logic, and would have to be added somewhat awkwardly to the atomic propositions or imposed from outside the logic.
For example, the agent might try to satisfy something implied by $\psi \land \varphi_2$, like $\exists (\eventually_{[1:2]}p)$ (i.e. there exists a path that satisfies $p$ within the next two states). 
However, at $m$ this is too permissive, since we really do prefer not changing lanes at all. 
And at $m'$ it is too restrictive, since $\varphi_0$ is a perfectly legitimate way of meeting $\alpha$'s obligations then.
Another method may be to specify behavior through reactive implications, e.g. "oncoming-traffic $\implies$ change lanes", but this sort of explicit rule must be built in by a human designer.
The conclusion is that there is a need to use a logic that captures preferences and derives obligations from them, as well as what agents are able and unable to do; a logic of agency \textit{and} obligation.

\section{Formalizing RSS in DAU}
\label{sec:rss in dau}
Responsibility-Sensitive Safety, or RSS, is a proposal put forth by Intel's Mobileye division~\cite{RSSv6}.
It proposes rules or requirements that, if followed by all cars in traffic, would lead to zero accidents.
Our objective here is to formalize some of the RSS rules in the language of Dominance Act Utilitarianism (DAU), and study their logical consequences.
Three important points must be made:
\begin{enumerate}[(A)]
	\item The formalization does not depend on the dynamical equations that govern the cars because we wish our conclusions to be independent of these lower-level concerns. 
	This is consistent with the standard AV control architecture where a logical planner decides what to do next (`change lanes' or `turn right') and a lower-level motion planner executes these decisions. 
	Our logical analysis concerns the logical planner.
	\label{note first}
	\item We are not trying to formalize general traffic laws or driving scenarios, which is outside the scope of this paper. We are only formalizing the RSS rules. 
	\item Every formalization, in any logic, can always be refined. We are not aiming for the most detailed formalization; we aim for a useful formalization.
	\label{note last}
\end{enumerate}
We have three objectives in doing so:
demonstrating the usefulness of DAU in a real use case; 
highlighting the ambiguities implicit in such proposals, which would go unnoticed without formalization;
and automating the checking of logical consistency and deriving of conclusions.
We first present the RSS rules in natural language (Section~\ref{sec:rss rules}), then their formalization (Section~\ref{sec:rss formalization}), and finally we analyze the rules' logical consequences. 

\subsection{The RSS rules}
\label{sec:rss rules}
The rules for Responsibility-Sensitive Safety are~\cite{RSSv6}:
\begin{enumerate}[RSS1.]
	\item Do not hit someone from behind. \label{rss-behind}
	\item Do not cut-in (to a neighboring lane) recklessly. \label{rss-cutin}
	\item Right-of-way is given, not taken. \label{rss-row}
	\item Be careful of areas with limited visibility. \label{rss-vis}
	\item If you can avoid an accident without causing another one, you must do it. \label{rss-avoid}
	\item To change lanes, you should not wait forever for a perfect gap: i.e., you should not wait for a gap large enough to get into even when the other car, already in the lane, maintains its current motion.\label{rss-assertive}
\end{enumerate}

RSS\ref{rss-assertive} is derived directly from the following in~\cite[Section 3]{RSSv6}: ``the interpretation [of the duty-of-care law] should lead to [...] an agile driving policy rather than an overly-defensive driving which inevitably would confuse other human drivers and will block traffic [...].
As an example of a valid, but not useful, interpretation is to assume that in order to be ``careful'' our
actions should not affect other road users. Meaning, if we want to change lane we should find a gap large enough such
that if other road users continue their own motion uninterrupted we could still squeeze-in without a collision. Clearly,
for most societies this interpretation is over-cautious and will lead the AV to block traffic and be non-useful.''
\textit{Note that, consistently with points (\ref{note first})-(\ref{note last}) above, this is stated without any reference to dynamics or specific scenarios. }
The RSS authors are concerned that overlay cautious driving might lead to unnatural traffic, so RSS aims to allow cars to move a bit assertively, and defines correct reactions to that.

%
We will not study RSS\ref{rss-vis} and \ref{rss-avoid} as they are currently too vague for formalization.

\subsection{Formalization of RSS Rules}
\label{sec:rss formalization}
\noindent\textbf{Formalizing RSS\ref{rss-behind}}. 
Let $\formula$ be a formula denoting `Hit someone from behind'.
A plausible formalization of RSS\ref{rss-behind} is then
\[RSS\ref{rss-behind}.\,\, \Ostit{\alpha}{\neg \formula}\]
That is, $\alpha$ ought to see to it that it does not hit anyone from behind.
However, suppose that $\alpha$ finds itself, through no fault of its own, in  a situation where a collision is unavoidable at time $m$, that is, $H_m = |\formula|_m^\Model$.
Then we can show that $RSS$\ref{rss-behind} cannot be met. 
This is something we know at design time.
There isn't much value in specifying obligations that remain in force even when they become impossible to meet, since we can't design controllers for them.
A better formalization of RSS\ref{rss-behind} would \textit{automatically, as a matter of logic}, remove the obligation when a collision becomes unavoidable.
This can be done using $dstit$ of Def.~\ref{def:dstit} as follows:
\[RSS\ref{rss-behind}r.\,\, \Ostit{\alpha}{\neg [\alpha \, dstit\!:\formula]}\]
This says that $\alpha$ should see to it that it does \textit{not} deliberately ensure an accident $\formula$.
This form of obligation is called \textit{refraining}: in this case, $\alpha$ \textit{refrains} from hitting anyone from behind.
$RSS\ref{rss-behind}$ and $RSS\ref{rss-behind}r$ \textit{are not logically equivalent}.
If $|\formula|_m^\Model = H_m$, then $\dstit{\alpha}{\formula}$ is necessarily false, and $RSS\ref{rss-behind}r$ is trivially satisfied since $\Ostit{\alpha}{\top}$ is a theorem of DAU.
Thus $RSS\!\ref{rss-behind}r$ does not impose unrealistic obligations on the agent.
Of course, a test engineer should then examine why the inevitable situation arose in the first place - but that is a separate debugging effort. 
The control engineer can now focus on designing a controller that meets the more realistic $RSS\!\ref{rss-behind}r$.
\\

\noindent\textbf{Formalizing RSS\ref{rss-cutin}}. 
Define two \CTLs~formulas, $\psi: $ a non-reckless cut-in, and $\psi_r$: a reckless cut-in.
Then RSS\ref{rss-cutin} is formalizable as 

\[RSS\ref{rss-cutin}. \Ostit{\alpha}{\forall \always (\psi\lor \psi_r \implies \neg \psi_r)}.\]

That is, $\alpha$ should see to it that always, if a cut-in happens, then it is a non-reckless cut-in.
\\

\noindent\textbf{Formalizing RSS\ref{rss-row}}.
Formalizing this rule requires some care. 
First, note that RSS\ref{rss-row} should probably be amended to say that `Right-of-way is given, not taken, \emph{and some car is given the right-of-way}' - otherwise, traffic comes to a standstill.
We will first focus on formalizing the prohibition (nobody should take the r-o-w), then we will formalize the positive obligation (somebody must be given it).

%
Let $\Agent = \{\alpha,\beta,\gamma,\ldots\}$ be a finite set of agents.
Define the atomic propositions
$GROW_\beta^\alpha$: $\beta$ gives right-of-way to $\alpha$ and
$p_\alpha$: $\alpha$ proceeds/drives through the conflict region.
Then
$TROW_\alpha \defeq p_\alpha \land \neg (GROW_\beta^\alpha \land GROW_\gamma^\alpha \land \ldots)$
formalizes taking the r-o-w: $\alpha$ proceeds without being given the right of way by everybody.
 We could now express the prohibition in RSS\ref{rss-row}: every $\alpha$ ought to see to it that it does not take the r-o-w:
 \begin{equation}
 \label{eq:rss-row prohib}
 RSS\ref{rss-row}prohib0. \, \bigwedge_{\alpha \in \Agent} \Ostit{\alpha}{\neg TROW_\alpha}
 \end{equation}

The difficulty with this formulation is that it could lead to $\alpha$ being obliged to force \textit{everybody else} to give it the r-o-w - something over which, a priori, it has no control. 
To see this, we need the following, whose proof is omitted due to lack of space. 
\begin{theorem}
	\label{thm:force others}
	Given obligations $A$ and $B$, 	$\Ostit{\alpha}{A \lor B} \land (\forall \neg A) \implies \Ostit{\alpha}{B}$
\end{theorem}
In other words, if $\alpha$ has an obligation to fulfill $A$ or $B$ at $m/h$, but every available history violates $A$ ($\forall \neg A$), then its obligation is effectively to fulfill $B$.
Applied to Eq.~\eqref{eq:rss-row prohib} with $A=\neg p_\alpha$ and $B=\land_{\beta\neq \alpha}GROW_\beta^\alpha$, Thm.~\ref{thm:force others} says that if $\alpha$ is in a situation where it has no choice but to proceed (e.g. as a result of slippage on a wet road, say), then its obligation is to see to it that everybody else gives it the right-of-way, which is unreasonable.

To remedy this, we first formalize the positive obligation: somebody must be given the right-of-way.
This seems to be a \textit{group obligation}: the \textit{group} must give r-o-w to one of its members.
Group obligations are formally defined in \cite[Ch. 6]{Horty01DLAgency}. 
Therefore, we define an atomic proposition $g_\alpha$: r-o-w is Granted to $\alpha$.
Then we formalize
 \begin{equation}
 \label{eq:rss-row pos}
 RSS\ref{rss-row}pos. \, \Ostit{\Agent}{\exists \lor_{\alpha \in \Agent} g_\alpha}
 \end{equation}
This says the group $\Agent$ has an obligation to give r-o-w to someone, and the only choice is in \textit{who} gets it.
We now come back to formalizing the prohibition:
  \begin{equation}
  \label{eq:rss-row prohib 2}
  RSS\ref{rss-row}prohib. \, \bigwedge_{\alpha \in \Agent} \Ostit{\alpha}{\always (\neg g_\alpha \implies \neg p_\alpha)}
  \end{equation}
  
Finally, we formalize $RSS\ref{rss-row}$ as the conjunction $RSS\ref{rss-row}prohib \,\land\, RSS\ref{rss-row}pos$.
\\

\noindent\textbf{Formalizing RSS\ref{rss-assertive}}.
This rule says that if the car wants to change lanes, it shouldn't wait for the perfect gap (otherwise, traffic is stalled).
First, let's formalize `waiting for the perfect gap', that is, waiting until the other car, already in the lane, gives the AV the right-of-way (e.g., by slowing down).
Let the atomic proposition $w_\alpha$ mean `$\alpha$ wants to change lanes' and recall that $p_\alpha$ means `$\alpha$ proceeds through the conflict region' while $g_\alpha$ means `$\alpha$ is Granted the right-of-way'.
For conciseness, let's introduce the bounded Release operator $\LTLbrelease_N$, which informally says that over the next $N$ steps, either $\psi$ does not hold at all, or it does and $\phi$ holds continuously until $\psi$ holds. 
\[\psi \LTLbrelease \phi = \psi \lor (\phi \land \Next \psi) \lor (\phi \land \Next \phi \land \Next^2 \psi) \lor \ldots\]
\[\ldots \lor (\phi \land \Next \phi \land \ldots \land  \Next^{N-1}\phi \land \Next^{N} \psi) \lor (\phi \land \Next \phi \land \ldots \Next^N \phi) \]
%
Then $\neg p_\alpha \LTLbrelease g_\alpha$ says that $\alpha$ waits for the perfect gap up to $N$ time steps (but we don't know what happens after this).
$\dstit{\alpha}{\neg p_\alpha \LTLbrelease g_\alpha}$ formalizes the agent deliberately seeing to it that it waits to be given the right-of-way, when it doesn't have to.
Finally,
\begin{equation}
\label{eq:rss-assertive}
RSS\!\ref{rss-assertive}.\quad \OstitC{\alpha}{\neg[\alpha~dstit:\neg p_\alpha \wbuntil g_\alpha]}{w_\alpha}
\end{equation}
formalizes that $\alpha$ ought to refrain from seeing to it that it waits for the right-of-way given that it wants to change lanes.
This obligation does not delay the lane change - in particular, it does not require the car to wait for the perfect gap.
It also does not rush $\alpha$: it can wait if it wishes to.
We emphasize that RSS assertive driving \textit{requires} that an AV sometimes force its way, as expressed in~\eqref{eq:rss-assertive}.
\subsection{Application: Inferring stit model structure}
\label{sec:analysis inferring policies}
%
\newcommand{\rssrow}{RSS\ref{rss-row}prohib}
\newcommand{\rssass}{RSS\!\ref{rss-assertive}}

In DAU, obligations are automatically derived from the stit model via Def.~\ref{def:cstit}.
Given an obligation that we \textit{want} the system to have, how should we structure the stit model so that it has that obligation?
This is similar to synthesis-from-specifications, an active research area in programming and in Cyber-Physical Systems.
This section gives an example where it is possible to manually partially infer the stit model structure from the RSS obligations.

Consider again the $\rssrow$ and $\rssass$ statements (Eqs.~\eqref{eq:rss-row prohib 2} and \eqref{eq:rss-assertive}).
 \begin{figure}[t]
 	\centering
 	\includegraphics[height=2.25in,keepaspectratio]{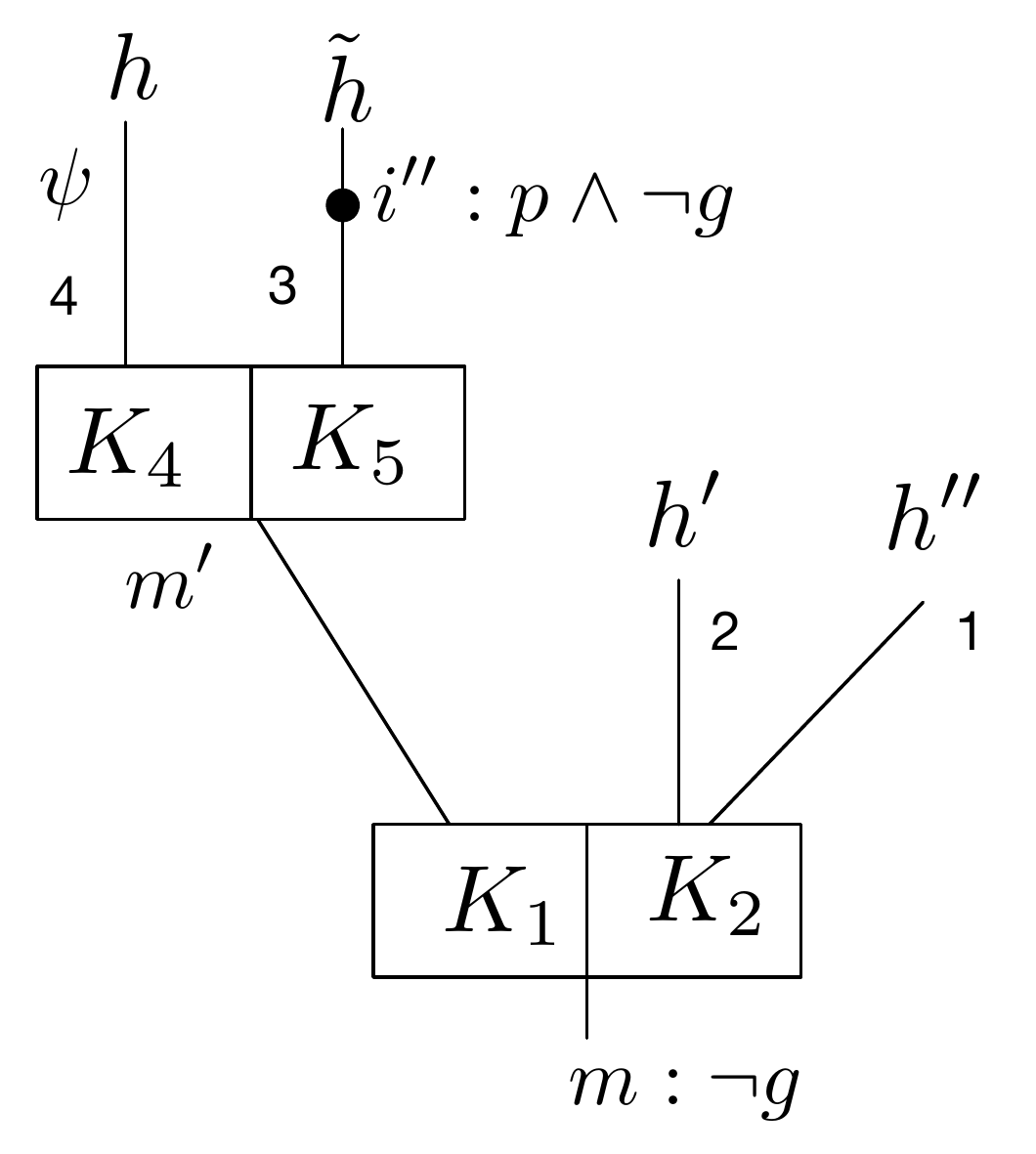}
 	\caption{The optimal action at $m$, $K_1$, necessarily contains history $\tilde{h}$ in which comes a moment $i''$ s.t. $i''/\tilde{h} \models p\land \neg g$. The controller must choose an action, prior to $i''$, that does not contain $\tilde{h}$. Since the controller always chooses optimal actions, the $\Value$ function must favor $K_4$, as shown.}
 	\label{fig:derive policies}
 \end{figure}
 \begin{proposition}
 	\label{prop:stit model inference}
 	A stit model has both obligations $\rssrow$ and $\rssass$ at $m$ if for every optimal action $K \in \Optimalam$, it holds that $|K|\geq 2$, and there exist a history $\tilde{h} \in K$ and a moment $m'>m$ in $\tilde{h}$ s.t.  $|\Choice{\alpha}{m'}|\geq 2$, $m/\tilde{h} \not\models_{\text{\CTLs}} \always(\neg g_\alpha \implies \neg p_\alpha)$ and $\tilde{h}$ is not in any optimal action at $m'$.
 \end{proposition}
The proof is omitted due to lack of space. 
The conclusion of the Proposition, illustrated in Fig.~\ref{fig:derive policies}, is counter-intuitive: it necessitates the existence of a history $\tilde{h}$ along which one of the formulas, $\always(\neg g_\alpha \implies \neg p_\alpha)$, is \textit{violated}.
But since the inferred structure places $\tilde{h}$ in a non-optimal action (via $\Value$), this doesn't lead to an \textit{obligation violation}.

\subsection{Application: undesirable consequence of RSS star-calculations}
\label{sec:other accidents}
One of the main tenets of RSS is that an AV is only responsible for avoiding potential accidents between itself and other cars (so-called `star calculations'); interactions between 2 other cars are not its concern~\cite[Remarks 1 and 8]{RSSv6}.
Yet everyday driving experience makes clear that our actions can be faulted for at least \textit{facilitating} an accident: e.g., by repeated braking, I may cause the car behind me to do the same, leading the car behind \textit{it} to rear-end it.
Or I might make a sudden lane change over two lanes, causing the car in the lane next to me to over-react when I speed past it, and collide with someone else.
We now show how this intuition is automatically captured by the DAU logic, and that RSS star-calculations lead to undesirable behavior of the AV.

Let $\formula \in $ \CTLs~denote a formula expressing ``Accident between two other cars'', and assume the accident is such that $\alpha$ can facilitate it as in the above 2 examples.
Then $\dstit{\alpha}{\formula}$ says that $\alpha$ (deliberately) sees to it that the accident happens even though it could avoid doing so;
given what we assumed about this accident, this means $\alpha$ facilitates the accident.
Then $\dstit{\alpha}{\neg \dstit{\alpha}{\formula}}$ expresses that $\alpha$ sees to it that it does \textit{not} facilitate the accident: this is a form of refraining.
Finally,
$\dstit{\alpha}{\neg \dstit{\alpha}{\neg \dstit{\alpha}{\formula}}}$ says that $\alpha$ refrains from refraining, that is, $\alpha$ does not refrain from facilitating the accident (even though it could). 
The RSS position is that it is OK for $\alpha$ to refrain from refraining~\cite[Remarks 1 and 8]{RSSv6}, as formalized here.

However, refraining from refraining is the same as doing. Formally~\cite[2.3.3.]{Horty01DLAgency}
\[ [\alpha \,dstit: \neg [\alpha\,dstit: \neg [\alpha\,dstit:\formula]]] \equiv \dstit{\alpha}{\formula}\]
And we argue that this matches our intuition: to not refrain from facilitating an accident even though one could is the same as facilitating it.
In other words, under this formalization, the RSS position is tantamount to allowing AVs to facilitate accidents between others - clearly, an undesirable conclusion.
This aspect of RSS, therefore, needs refinement to take into account longer-range interactions between traffic participants.

\section{System Design and Model Checking DAU Obligations}
\label{sec:model checking}
The system designer's job is to design a system that has the \textit{right} obligations;
it is then the control engineer's job to design a controller that makes the system meet these obligations.
In DAU, obligations are automatically derived from stit models/trees, but designers usually model an agent as an automaton or a similar structure.
The question then naturally poses itself: given an agent model, how do we verify whether it has a given obligation?
Answering this question is a crucial design step: there is no point designing controllers that meet the wrong obligations. 
%
This can be cast as a model-checking question, which this section tackles.
All proofs are in the appendices.
\begin{figure*}[t]
	\centering
	\includegraphics[height=2.25in,keepaspectratio]{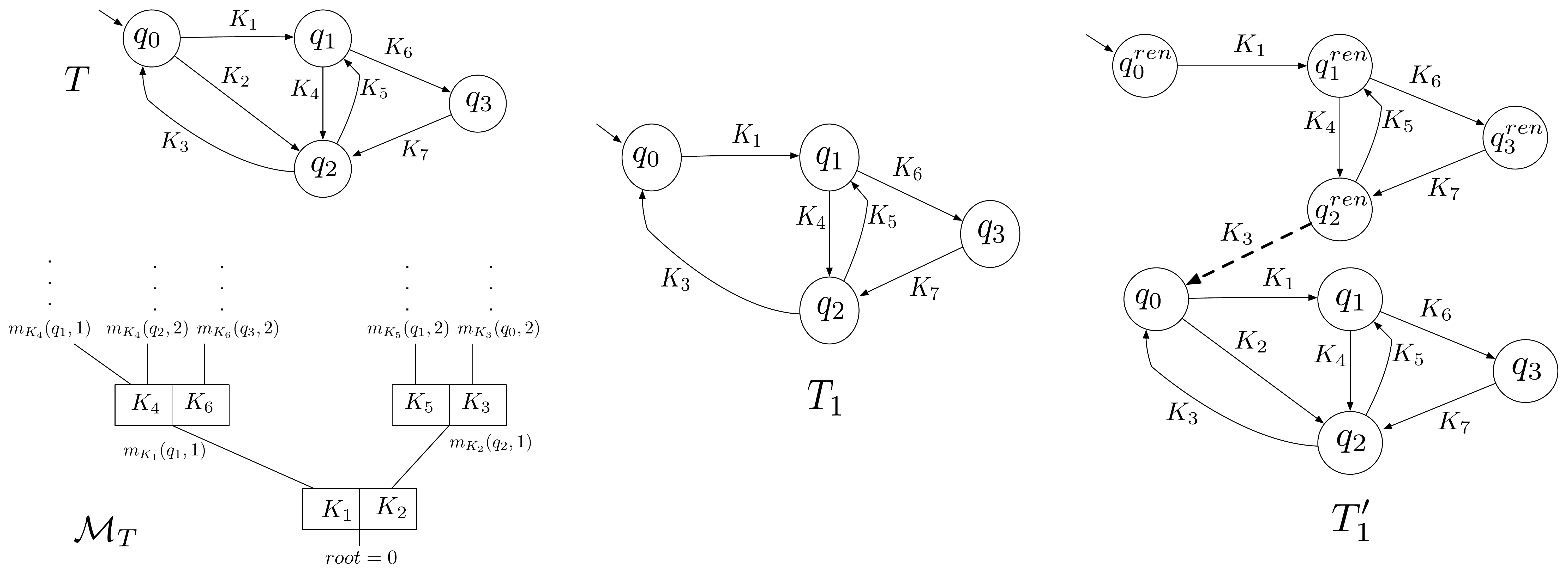}
	\caption{Left: a stit model generated by executing the stit automaton $T$ (transition weights not shown). Center and right: Automata $T_n$ and $T_n'$ used in Algorithm~\ref{algo:mc}. $T_1$ only has $K_1$ as first action, and $T_1'$ is obtained by re-naming states of $T_1$ and `adding’ a copy of $T_1$ to it. Executions of $T_1'$ are simply the execution of $T$ that start with $K_1$.}
	\label{fig:algo}
\end{figure*}

\subsection{Modeling an agent}
\label{sec:modeling an agent}
\begin{definition}[Stit automaton]
	\label{def:weighted automaton}
	Let $AP$ be a finite set of atomic propositions.
	A \emph{stit automaton} $T$ is a tuple $T = (Q, q_0, \Act, F, \Transrel, L, w,\lambda)$, where
	$Q$ is a finite set of states,
	$q_0$ is the initial state,
	$\Act$ is a finite set of actions ($\Act \subset 2^{H_{root}})$), 
	$F \subset Q$ is a set of final states,
	$\Transrel \subset Q\times \Act \times Q$ is a finite transition relation such that if $(q,K,q')$ and $(q,K',q')$ are in $\Transrel$ then $K=K'$,
	$L:Q\rightarrow 2^{AP}$ is a labeling function,
	$w: \Transrel \rightarrow \Re$ is a weight function,
	and $\lambda: \Re^\omega \rightarrow \Re$ is an accumulation function.
\end{definition}
Denote by $\Transrel(q) \subset \Transrel$ the set of outgoing transitions from $q$ ($\Transrel(q) = \{(q,K,q') \in \Transrel\}$), 
by $Post(q,K) = \{q' \such (q,K,q') \in \Transrel(q)\}$ the successors of $q$ under $K$,
and by $Post(q) = \cup_{K:(q,K,q') \in \Transrel(q)} $
$Post(q,K)$ all the successors of $q$.
Finally, we denote by $T.q_0$ the initial state of $T$ when there's a need to clarify the automaton.
Note that $T$ is a type of non-deterministic weighted automaton. 
Its unweighted counterpart $T^u$ is a classical transition system; 
\textit{thus for a \CTLs~formula $\formula$, we could model-check whether $T^u\models \formula$.}
A set of agents is modeled by the product of all individual stit automata, which is itself a stit automaton.
(When taking the product, we must define how weights are combined and how to construct the product's accumulation function, which are application-specific considerations.)  
Therefore the rest of this section applies to stit automata, whether they model one or multiple agents. 
We will continue to refer to one agent $\alpha$ for simplicity.

\begin{algorithm}
	\DontPrintSemicolon
	\KwData{A stit automaton $T = (Q, q_0, \Act, F, \Transrel, L, w,\lambda)$,  an obligation $A$}
	\KwResult{$\Model_T,root/h \models \Ostit{\alpha}{A}$}
	Set $root=0$\;
	Set $\Choicear=\{K \in \Act \such (q_0,K,q')\in \Delta \text{ for some }q'\} =\{K_1,\ldots,K_m\}$ \;
	\tcp*[l]{First step: find optimal actions at $root$}
	\For{$1\leq n \leq m$}{
		\tcc*[l]{Construct automaton $T_n'$ s.t. every execution of $T_n'$ is an execution of $T$ starting with action $K_n$. See Fig.~\ref{fig:algo}.}
		Create automaton $T_n$ by deleting all transitions $(q_0,K,q')$ with $K\neq K_n$ \;
		Create a copy $T_n^{\text{ren}}$ of $T_n$\;
		Create the automaton $T_n'$ as a union of $T_n^{\text{ren}}$ and $T$, with every transition $(q,K,T_n^{\text{ren}}.q_0)$ in $T_n^{\text{ren}}$ replaced by a transition $(q,K,T.q_0)$\;
		\lnl{comp vn} Compute the max value, $\vnu$, and min value, $\vnl$, of any $T_n'$ strategy starting at $q_0$\;
	}
	\tcc*[l]{An interval $ [\vnl,\vnu] $ is \textit{un-dominated} if there is no other interval $[\vnl', \vnu']$, computed in the above for-loop, s.t. $\vnl' > \vnu$}
	\lnl{undomin}Find all un-dominated intervals $[\vnl,\vnu]$\;
	\lnl{line:optimalam}Set $\Optimalar = \{K_n \in \Choicear \such [\vnl,\vnu] \text{ is un-dominated}\}$\;
	\tcc*[l]{Second step: decide whether all actions $K$  in $\Optimalar$ guarantee $A$, i.e., $K\subseteq |A|_{root}$.}
	\lnl{line:for}\For{$K_n \in \Optimalar$}{
		\uIf{$A$ is a \CTLs~formula}
		{
			\tcc{Does every execution of $T$ starting with $K_n$ satisfy $A$?}
			Use \CTLs~model-checking to check whether $T_n' \models_{\text{\CTLs}}\forall A$\;
			\If{$T_n' \not\models_{\text{\CTLs}} \forall A$ \tcp*{Optimal action $K_n$ does not guarantee $A$}}
			{\lnl{line:false 1}return False}
		}
		\lnl{line:dstit pos} \uElseIf{$A = \dstit{\alpha}{\formula}$ with $\formula\in$ \CTLs}
		{
			\tcp*[l]{This is true iff $H_{root}=|\formula|_{root}$}
			\lnl{line:all phi}Model-check whether $T\models_{\text{\CTLs}}\forall \formula$\;
			\tcc*[l]{This is true iff $K_n$ guarantees $\formula$, is not equiv. to line~\ref{line:all phi}}
			Model-check whether $T_n'\models_{\text{\CTLs}}\forall \formula$\;
			\If{$T\models_{\text{\CTLs}}\forall \formula$ or $T_n'\not \models_{\text{\CTLs}}\forall \formula$}
			{					
				\lnl{line:false 2}return False
			}
		}
		\lnl{line:dstit neg}\Else{
			\tcc{Last case: $A = \neg \dstit{\alpha}{\formula}$ with $\formula\in$ \CTLs. Similar to previous case on line \ref{line:dstit pos} with obvious modifications} 
		}
	} 
	\lnl{line:true}Return True\;
	\caption{Model checking DAU.}
	\label{algo:mc}
\end{algorithm}

\paragraph{From automata to stit models.}
Let $S^\omega$ denote the set of infinite sequences $(a_i)_{i\in \Ne}$ with $a_i\in S$.
An \textit{execution} of a stit automaton $T$ is a sequence $\pi \in \Transrel^\omega$ of transitions of the form $\pi = (q_0,K_0,q_1)(q_1,K_1,q_2)\ldots$.
The corresponding sequence of actions $K_0,K_1,\ldots \in \Act^\omega$ is called a \textit{strategy}.
Because of non-determinism, a strategy can produce multiple executions.
%
%
An execution of the automaton generates a stit model in the natural way: starting in state $q_0$ the automaton takes an infinite sequence of actions from $\Act$, thus non-deterministically traversing an infinite number of transitions $e$ from $\Transrel$.
These sequences of transitions form the histories in the corresponding stit model, with every transition $e$ adding a moment to the histories. 
The value(s) of those histories are obtained by accumulating $w(e)$ along the traversed transitions using function $\lambda$.
See Fig.~\ref{fig:algo} for an example.
The formal construction and proof are in Appendix~\ref{apx: proof of Mt}.
\begin{theorem}
	\label{prop:Mt}
	The structure $\Model_T$ obtained by executing a stit automaton $T$ is a utilitarian stit model with finite $\Choiceam$ for every agent $\alpha$ and moment $m$.
\end{theorem}

\subsection{Model checking algorithm}
\label{sec:mc algo}
The \emph{cstit model-checking problem} is:
Given a stit automaton $T$ that models an agent $\alpha$ and an obligation $A$, determine whether
$\Model_T,root/h \models \Ostit{\alpha}{A}$ for some $h \in H_{root}$. The case of conditional oughts $\OstitC{\alpha}{A}{B}$ is similarly handled and we omit the details.

Given the structure of an obligation given in Def.~\ref{def:obligation structure}, the model-checking problem can be broken down into two parts: what is the set of optimal actions at $root$, $\Optimalar$? 
And out of these optimal actions, which ones guarantee the truth of $A$?
(Recall Eqs.~\eqref{eq:proposition dominance}-\eqref{eq:optimalam}:  action optimality is determined solely by the $\Value$ function, and not by which obligations its histories satisfy).
If all optimal actions guarantee $A$, then by Def.~\ref{def:cstit}, $\Model_T$ has obligation $A$ at $root/h$.
The algorithm is presented in Algorithm~\ref{algo:mc} page~\pageref{algo:mc}.
In it, $\models_{\text{\CTLs}}$ denotes the classical \CTLs~satisfaction relation.
\begin{theorem}
	\label{thm:mc is possible}
	Algorithm~\ref{algo:mc} returns True iff $\Model,root/h\models \Ostit{\alpha}{A}$. 
	It has complexity $O(2m(|T|+c_\lambda+|T|\cdot 2^{|\formula|}))$, where $c_\lambda$ is the cost of computing the minimum and maximum values of a strategy executed on automaton $T$ and $|T|$ is the number of states and transitions in $T$.
\end{theorem}

The proof is in Appendix~\ref{apx:proof of mc}. 
This algorithm can be amended to accept a conditional obligation $\OstitC{\alpha}{A}{B}$ by accepting only those actions $K$ in $\Optimalar$ that guarantee $A$ and $B$.
The computation of the minimum and maximum values of a strategy's execution line~\ref{comp vn} clearly depends on the function $\lambda$ used for accumulating weights along the execution: e.g., if $\lambda$ is addition and all the weights are positive, then all executions have infinite value, and every future is ideal, which is a comforting thought but of little interest in modeling the real world.
This question is related to but distinct from temporal logic accumulation~\cite{Boker14Accumulative} and quantitative languages~\cite{Chatterjee08Quantitative}.
We give now one example of a $\lambda$ that can model real-world phenomena, and lead to finite values of $\vnu$.
Take $\lambda = \min$.
For instance, if $w((q,K,q'))$ is the time-to-collision resulting from action $K$ then $\Value(h)$ is the shortest time-to-collision encountered along the history, and an optimal history is one with the highest minimum time-to-collision.
It's a simple matter to prove that $\vnu$ is the maximum weight of any reachable transition from $q_0$, which can be computed in a finite number of steps.
(Unfortunately, different $\lambda$s will, in general, require different customized analyzes.)


\section{Conclusions}
\label{sec:conclusion}
We have demonstrated the use of Dominance Act Utilitarianism in formalizing safety norms for autonomous vehicles.
Our objective was to assess the feasibility and utility of doing so: we expressed safety norms from RSS in DAU; 
found undesirable consequences in these norms;
and showed that system designers can automatically derive a formalized system's obligations and objectives.

It is desirable next to enrich the interaction between deontic and temporal modalities, e.g. to express things like `In the next planning cycle the AV must see to it that it changes lanes'.  
This then allows reasoning about obligation propagation through time~\cite{Brunel08Propagation}.
It will be equally important to study \textit{obligation inheritance} between groups and individuals: e.g., if it is the group's obligation to give the right-of-way, what does that imply for individual obligations? Given that deontic logics were developed for ethical analysis, this work also opens the way to formally considering ethical implications of system design. In our experience even framing technical specifications as obligations can make explicit an implicit norm. Addressing ethical considerations is necessary to build trust in autonomous systems, and this work suggests it may be possible to formalize a a system's ethical constraints, and analyze the moral implications of its design.
These and other considerations will ultimately determine the suitability of DAU for AV design and verification. 

\appendix

\section{More elements of Dominance Act Utilitarianism}
\label{apx:choice}
\textbf{Agent choice}.
The choice mapping $\Choiceam$ in a general deontic stit model obeys
\begin{itemize}
	\item The actions in $\Choiceam$ partition the set $H_m$: $K\cap K' = \emptyset$ for every $K,K'$ and $\cup_{K\in \Choiceam} K = H_m$.
	There is no loss of generality in this constraint, it is a formality that allows us to maintain the useful tree structure.
	\item Independence of agents: given any group of agents $\Group \subseteq \Agent$, $\cap_{\alpha\in \Group}\Choiceam  \neq \emptyset$.
	That is, the actions of one agent do not prevent the choice of action available to any other agent at the \textit{same} moment $m$. 
	\item No choice between undivided histories: If two histories are still undivided at $m$ (that is, they share a moment $m'>m$) then they belong to the same action $K$ in $\Choiceam$.
\end{itemize}

\section{Construction of $\Model_T$ and Proof of Thm.~\ref{prop:Mt}}
\label{apx: proof of Mt}
We give the formal construction of stit model $\Model_T$ from stit automaton $T$, then prove Thm.~\ref{prop:Mt}.
The construction is as follows (see Fig.~\ref{fig:algo}).

$\bullet$ \textit{Initialization}: set iteration $i=1$, $q = q_0$, $root=0$, $S=\{\langle q_0, root\rangle \}$, $\Tree = \{root\}$.

$\bullet$ \textit{Expansion}: 
Set $S' = \emptyset$. 
For every couple $\langle q, m\rangle \in S$, 
\begin{enumerate}[Exp1)]
	\item set $\Choiceam = \{K: (q,K,q')\in \Transrel(q) \text{ for some } q'\}$: the agent has a choice of actions at $m$ from the actions that label the transitions out of $q$. \label{expand:1}
	\item For every $K\in \Choiceam$, and every $q' \in Post(q,K)$, add a new moment $m_K(q',i)$ to $\Tree$ with $m_K(q',i) > m$, and such that the history ending with the moments $( m,m_K(q',i))$ belongs to action $K$. 
	Also, add the couple $\langle q',m_K(q',i)\rangle$ to $S'$. \label{expand:2}
	\item Set the label map $v(m/h) = L(q)$ for every history $h$ passing through $m$. \label{expand:3}
\end{enumerate}

$\bullet$ \textit{Update}: Set $S=S'$.
For the next iteration, set $i=i+1$. 
Goto \textit{Expansion}.

$\bullet$ \textit{Valuation}: For every history $h$ constructed in the \textit{Expansion} loop, its value is computed as $\Value(h)=\lambda(w(e_i)_{i\in \Ne})$ where $e_i$'s are the transitions taken while constructing $h$.
($\lambda$ must be such that infinite accumulation yields a finite value).

\begin{proof}[Thm.~\ref{prop:Mt}]
We first verify that $\Model_T$ is a branching time model (Def.~\ref{def:frame and model}). 
The ordering between moments is irreflexive and transitive by construction. 

Take 3 moments $m_1$, $m_2$ and $m_3$ s.t. $m_1<m_3$ and $m_2<m_3$. 
Moments are only added in Exp\ref{expand:2} so $m_3=m_K(q',i)$ for some $q',i$, and by construction there is a unique moment $m_{K'}(q,i-1)$ at level $i-1$ s.t. $m_{K}(q',i)>m_{K'}(q,i-1)$.
By a simple inductive argument, there is a unique moment $m_{K_j}(q_j,j)$ at level $j$ s.t. $m_K(q',i)>m_{K_j}(q_j,j)$ for \textit{every} $j<i$. 
Thus the sequence of moments that are smaller than $m(q',i)$ forms a chain (a linear order) to which must belong both $m_2$ and $m_3$, so either $m_2<m_3$ or $m_3<m_2$.

The tree is rooted at $0$ as can be easily established by induction on $i$.

The function $v$ in Exp\ref{expand:3} plays the role of the stit model's label map.

We now show that $\Choiceam$ satisfies the constraints of Appendix~\ref{apx:choice} on choices:
\\
$\bullet$ The actions in $\Choiceam$ partition $H_m$: indeed, take a history starting at $m=m_K(q,i-1)$. It is expanded in Exp\ref{expand:2} only, by $m_{K'}(q',i)$ say, and the expanded history $\langle m,m_{K'}(q',i)\rangle$ is assigned to only one action. 
Thus the histories $\langle m,m_{K'}(q',i)\rangle$, $q'\in Post(q,K'), K'\in \Choiceam$ are partitioned among the actions at $m$.
By definition of the automaton transition relation, two different actions must lead to two different states $q'$ and $q''$ so the newly created moments $m_{K'}(q',i+1)$ and $m_{K''}(q'',i+1)$ at the next iteration $i+1$, and which expand these histories, are different.
Therefore, two histories that were in different actions at $m$ will never share a moment after $m$.
Thus the actions at $m$ partition $H_{m}$.
\\
$\bullet$ Independence of agents: this is automatically guaranteed by using an automaton that models the product of all stit automata.
\\
$\bullet$ No choice between undivided histories: as established in the first bullet of the proof, histories that are in different actions at $m$ will never share a moment after $m$.
Therefore, two histories that share a moment at $m'>m$ must be in the same action at $m$.

Finally, $\Choiceam$ is finite for each moment since, as can be seen in Exp\ref{expand:1}, $\Choiceam$ is (isomorphic to) a subset of $\Delta$ and the latter is finite.
\hspace*{\fill}QED.
\end{proof}

\section{Proof of Thm.~\ref{thm:mc is possible}}
\label{apx:proof of mc}
Recall that by executing a stit automaton, a stit model is created (Appendix~\ref{apx: proof of Mt}).

\begin{lemma}
	\label{lemma:Tn'}
	The histories generated by $T_n'$ are exactly the histories of $T$ whose first action is $K_n$, modulo a re-naming of the states.
\end{lemma} 
\begin{proof}
	Recall that $T_n'$ has two components, namely a copy $T_n^{\text{ren}}$ of $T_n$ and a copy of $T$. 
	See Fig.~\ref{fig:algo}.
	$T_n$ is obtained by removing transitions from $T$, thus every history generated by $T_n$ is a valid $T$-history.
	Every history generated by $T_n$ starts with $K_n$ by construction.
	So every history $h$ of $T_n'$ starts with $K_n$, because it starts in $T_n^{\text{ren}}$.
	
	\underline{Case 1: $h$ never leaves $T_n^{\text{ren}}$}.
	$T_n^{\text{ren}}$ is nothing but a renaming of $T_n$ and we've already established that a history of $T_n$ is a history of $T$, so this case is done.
	
	\underline{Case 2: $h$ leaves $T_n^{\text{ren}}$}.
	That is, a transition takes the execution into the $T$ copy. 
	Up to the transition, $h$ is a history of $T$ as established in Case 1. 
	The transition itself, say $(q, K,T.q_0)$, is a valid transition of $T$ (modulo re-naming) since it was created by replacing a $T$ transition of the form $(T.q,T.K, T.q_0)$.
	Once in the $T$ copy, the history of course continues to be a valid history of $T$.
	\hspace*{\fill}QED.
\end{proof}

\begin{lemma}
	\label{lemma:compute optimalam}
	The set computed at line~\ref{line:optimalam} is indeed $\Optimalar$.
\end{lemma} 
\begin{proof}
	Every history of $T_n'$ starts with $K_n$ so $\vnl = \min\{\Value(h) \such h \in K_n \}$ and $\vnu=\max \{\Value(h) \such h\in K_n\}$.
	By definition of action dominance, $K_n \preceq K_n'$ in $T$ iff $\vnu\leq \vnl'$.
	So $[\vnl,\vnu]$ is un-dominated iff its action $K_n$ is un-dominated and must be optimal.
	\hspace*{\fill}QED.
\end{proof}

\begin{lemma}
	\label{lemma:false 1}
	If line~\ref{line:false 1} is executed, then $K_n \nsubseteq |A|_{root}$.
\end{lemma} 
\begin{proof}
If $T_n'\not\models \forall A$ this means some execution $\tilde{h}$ of $T_n'$ violates $A$.
By Lemma~\ref{lemma:Tn'} $\tilde{h}$ is also a history of $T$ starting with the optimal action $K_n$, so that $K_n \nsubseteq |A|_{root}$. 
QED.
\end{proof}

\begin{lemma}
	\label{lemma:false 2}
	If line~\ref{line:false 2} is executed, then $\Model,root/h\not\models \Ostit{\alpha}{\formula}$
\end{lemma} 
\begin{proof}
	$T \models_{\text{\CTLs}} \forall \formula$ iff every history of $T$ satisfies $\formula$ and so $H_{root}=|\formula|_{root}$; in this case, by definition of $dstit$, $root/h \not\models \dstit{\alpha}{\formula}$.
	$T_n'\not\models_{\text{\CTLs}} \forall \formula$ iff there exists a history $\tilde{h}$ of $T_n'$ which violates $\formula$.
	Again this is also a history of $T$ which belongs to the optimal $K_n$ so that $K_n \nsubseteq |\formula|_{root}$.
\hspace*{\fill}QED.
\end{proof}

\begin{proof}[Thm.~\ref{thm:mc is possible}]
	We need to establish that the algorithm returns True iff $K\subseteq |A|_{root}$ for every optimal $K$.
	The set of optimal actions is computed at line~\ref{line:optimalam} by Lemma~\ref{lemma:compute optimalam}.
	The for-loop at line~\ref{line:for} visits each optimal action in turn. 
	Line~\ref{line:true} is executed iff none of the `return False' statements preceding it are executed;
	namely, iff $K \subseteq |A|_{root}$ by Lemma~\ref{lemma:false 1} in Case $A$ is \CTLs,
	or iff $H\neq |\formula|_{root}$ and $K\subseteq |\formula|_{root}$ in the case of line~\ref{line:dstit pos} by Lemma~\ref{lemma:false 2} 
	(and the case of line~\ref{line:dstit neg} is similarly treated).
	These are the definition of $root/h\models \Ostit{\alpha}{A}$.
	
	For the complexity, the first for-loop takes $2|T|$ operations per iteration to create the automata copies and $2c_\lambda$ to compute $\ell_n$ and $u_n$.
	Finding the un-dominated intervals takes $m-1$ comparisons to find the largest $\ell_n$ and $m$ to compare each $u_n$ to $\max \ell_n$.
	The second for-loop does at the most two CTL$^*$ model-checking runs per optimal action; each run has complexity $O(|T|\cdot 2^{|\formula|})$ and there are at most $m$ optimal actions.
	The total is then $O(2m(|T|+c_\lambda) + 2m-1 + 2m(|T|2^{|\formula|}))$.
\hspace*{\fill}QED.
\end{proof}

%
%
%
\bibliographystyle{ACM-Reference-Format}
 \bibliography{iccps2017,hscc17,hscc2016,fainekos_bibrefs,hscc19,cav2019,colin_bib}
\end{document}